\documentclass[journal, onecolumn]{IEEEtran}
\usepackage[T1]{fontenc}
\usepackage{amssymb,amsmath,amsthm}
\usepackage[colorlinks=true]{hyperref}
\usepackage{url}
\usepackage{breakurl}
\usepackage{verbatim}
\usepackage{ifthen}
\usepackage{framed}
\usepackage{multirow}
\usepackage{hhline}
\newtheorem{df}{Definition}[section]
\newtheorem{lem}{Lemma}[section]

\newtheorem{thm}{Theorem}[section]

\DeclareMathOperator{\tr}{tr}
\DeclareMathOperator{\id}{id}

\DeclareMathOperator{\hmax}{H_{max}}
\DeclareMathOperator{\hmin}{H_{min}}
\DeclareMathOperator{\dham}{d_{H}}

\newcommand{\ket}[1]{| \hspace{1pt} #1 \rangle}

\newcommand{\ketbrad}[2]{| \hspace{1pt} #1 \rangle \langle #2 \hspace{1pt} |}

\newcommand{\ketbra}[1]{\ketbrad{#1}{#1}}
\newcommand{\bramatket}[3]{\langle #1 \hspace{1pt} | #2 | \hspace{1pt} #3 \rangle}
\newcommand{\nbox}[2]{\hspace{#2pt} \mbox{#1} \hspace{#2pt}}
\newcommand{\norm}[2][]{#1| \! #1| #2 #1| \! #1|}
\newcommand{\abs}[2][]{#1| #2 #1|}
\def \includeprotocol {0}
\def \includetikz {0}

\newcommand{\cB}{\mathcal{B}}

\newcommand{\cH}{\mathcal{H}}

\newcommand{\cK}{\mathcal{K}}

\newcommand{\cS}{\mathcal{S}}

\newcommand{\cU}{\mathcal{U}}
\newcommand{\cV}{\mathcal{V}}

\newcommand{\cX}{\mathcal{X}}
\newcommand{\cY}{\mathcal{Y}}
\newcommand{\cZ}{\mathcal{Z}}

\def \includeprotocol {1}
\def \includetikz {1}
\ifthenelse{\equal{\includeprotocol}{1}}
{
\newcounter{protcounter}
\setcounter{protcounter}{0}
\newenvironment{prot}[1]
{
	\refstepcounter{protcounter}
	\begin{framed}
	\noindent \textbf{Protocol~\arabic{protcounter}:}\ {\texttt{#1}}\\
}
{
	\end{framed}
}
}
{}

\ifthenelse{\equal{\includetikz}{1}}
{
\usepackage{tikz}
\usepackage{subfig}
\usetikzlibrary{arrows, calc}
}
{}

\newcommand*{\splitmodelneither}[2]
{
\begin{tikzpicture}[scale=#1]
	\draw[fill=gray!30] (2, 0) rectangle (4, 1);
	\draw[fill=gray!30] (2, 2.5) rectangle (4, 3.5);
	\node (alice) at (3, 3) {Alice};
	\node (bob) at (3, 0.5) {Bob};
	{
	\ifthenelse{\equal{#2}{empty}}
	{}
	{\node (phase) at (3, 4.5) {#2};}
	}
	\draw [thick, ->] (3.15, 2.4) -- (3.15, 1.1);
	\draw [thick, ->] (2.85, 1.1) -- (2.85, 2.4);
\end{tikzpicture}
} 
\newcommand*{\splitmodelAlice}[2]
{
\begin{tikzpicture}[scale=#1]
	\draw[fill=black] (2.95, 2) rectangle (3.05, 3.5);
	\draw[fill=gray!30] (2, 0) rectangle (4, 1);
	\draw[fill=gray!30] (-0.25, 2.5) rectangle (2.25, 3.5);
		\draw[fill=gray!30] (3.75, 2.5) rectangle (6.25, 3.5);
	\node (alice) at (1, 3) {Alice};
	\node (amy) at (5, 3) {Amy};
	\node (bob) at (3, 0.5) {Bob};
	{
	\ifthenelse{\equal{#2}{empty}}
	{}
	{\node (phase) at (3, 4.5) {#2};}
	}
	\draw [thick, ->] (0.85, 2.4) -- (2.55, 1.1);
	\draw [thick, ->] (2.85, 1.1) -- (1.15, 2.4);
	\draw [thick, ->] (4.85, 2.4) -- (3.15, 1.1);
	\draw [thick, ->] (3.45, 1.1) -- (5.15, 2.4);
\end{tikzpicture}
}
\newcommand*{\splitmodelBob}[2]
{
\begin{tikzpicture}[scale=#1]
	\draw[fill=black] (2.95, 0) rectangle (3.05, 1.5);
	\draw[fill=gray!30] (2, 2.5) rectangle (4, 3.5);
	\draw[fill=gray!30] (-0.25, 0) rectangle (2.25, 1);
	\draw[fill=gray!30] (3.75, 0) rectangle (6.25, 1);
	\node (alice) at (3, 3) {Alice};
	\node (bob) at (1, 0.5) {Bob};
	\node (brian) at (5, 0.5) {Brian};
	{
	\ifthenelse{\equal{#2}{empty}}
	{}
	{\node (phase) at (3, 4.5) {#2};}
	}
	\draw [thick, ->] (0.85, 1.1) -- (2.55, 2.4);
	\draw [thick, ->] (2.85, 2.4) -- (1.15, 1.1);
	\draw [thick, ->] (4.85, 1.1) -- (3.15, 2.4);
	\draw [thick, ->] (3.45, 2.4) -- (5.15, 1.1);
\end{tikzpicture}
}
\newcommand*{\splitmodelboth}[2]
{
\begin{tikzpicture}[scale=#1]
	\draw[fill=black] (2.95, 0) rectangle (3.05, 3.5);
	\draw[fill=gray!30] (-0.25, 2.5) rectangle (2.25, 3.5);
	\draw[fill=gray!30] (3.75, 2.5) rectangle (6.25, 3.5);
	\draw[fill=gray!30] (-0.25, 0) rectangle (2.25, 1);
	\draw[fill=gray!30] (3.75, 0) rectangle (6.25, 1);
	\node (alice) at (1, 3) {Alice};
	\node (amy) at (5, 3) {Amy};
	\node (bob) at (1, 0.5) {Bob};
	\node (brian) at (5, 0.5) {Brian};
	{
	\ifthenelse{\equal{#2}{empty}}
	{}
	{\node (phase) at (3, 4.5) {#2};}
	}
	\draw [thick, ->] (0.85, 2.4) -- (0.85, 1.1);
	\draw [thick, ->] (1.15, 1.1) -- (1.15, 2.4);
	\draw [thick, ->] (4.85, 2.4) -- (4.85, 1.1);
	\draw [thick, ->] (5.15, 1.1) -- (5.15, 2.4);
\end{tikzpicture}
}
\newcommand{\splits}[6]
{
\begin{figure}[h]
\centering
\subfloat
{
\ifthenelse{\equal{#1}{1}} {\splitmodelneither{#4}{Commit phase}} {}
\ifthenelse{\equal{#1}{2}} {\splitmodelAlice{#4}{Commit phase}} {}
\ifthenelse{\equal{#1}{3}} {\splitmodelBob{#4}{Commit phase}} {}
\ifthenelse{\equal{#1}{4}} {\splitmodelboth{#4}{Commit phase}} {}
}
\hspace{20pt}
\subfloat
{
\ifthenelse{\equal{#2}{1}} {\splitmodelneither{#4}{Wait phase}} {}
\ifthenelse{\equal{#2}{2}} {\splitmodelAlice{#4}{Wait phase}} {}
\ifthenelse{\equal{#2}{3}} {\splitmodelBob{#4}{Wait phase}} {}
\ifthenelse{\equal{#2}{4}} {\splitmodelboth{#4}{Wait phase}} {}
}
\hspace{20pt}
\subfloat
{
\ifthenelse{\equal{#3}{1}} {\splitmodelneither{#4}{Open phase}} {}
\ifthenelse{\equal{#3}{2}} {\splitmodelAlice{#4}{Open phase}} {}
\ifthenelse{\equal{#3}{3}} {\splitmodelBob{#4}{Open phase}} {}
\ifthenelse{\equal{#3}{4}} {\splitmodelboth{#4}{Open phase}} {}
}
{
\ifthenelse{\equal{#5}{empty}}
	{}
	{\caption{#5\label{#6}}}
}
\end{figure}
}

\begin{document}
\title{Secure Bit Commitment From Relativistic Constraints}
\author{J\k{e}drzej~Kaniewski,~Marco Tomamichel,~Esther~H{\"a}nggi~and~Stephanie~Wehner
\thanks{Manuscript received xxx; revised yyy. This work is supported by the National Research Foundation and the Ministry of Education of Singapore. Copyright (c) 2012 IEEE. Personal use of this material is permitted.  However, permission to use this material for any other purposes must be obtained from the IEEE by sending a request to pubs-permissions@ieee.org.}
\thanks{J.~Kaniewski, M.~Tomamichel, E.~H{\"a}nggi and S.~Wehner are with the Centre for Quantum Technologies, National University of Singapore, 3 Science Drive 2, Singapore 117543, email:~j.kaniewski@nus.edu.sg.}}
\maketitle
\begin{abstract}
We investigate two-party cryptographic protocols that are secure under
assumptions motivated by physics, namely special relativity and quantum mechanics. In particular, we discuss the security of bit commitment in so-called split models, i.e.~models in which at least one of the parties is not allowed to communicate during certain phases of the protocol. We find the minimal splits that are necessary to evade the Mayers-Lo-Chau no-go argument and present protocols that achieve security in these split models.
Furthermore, we introduce the notion of local versus global command, a subtle issue that arises when the split committer is required to delegate non-communicating agents to open the commitment. We argue that classical protocols are insecure under global command in the split model we consider. On the other hand, we provide a rigorous security proof in the global command model for Kent's quantum protocol~\cite{kent12a}. The proof employs two fundamental principles of modern physics, the no-signalling property of relativity and the uncertainty principle of quantum mechanics.
\end{abstract}
\section{Introduction}
\label{sec:introduction}
\IEEEPARstart{T}{he} goal of two-party cryptography is to enable two parties, Alice and Bob, to solve a task in cooperation even if they do not trust each other. An example of such a task is the cryptographic primitive known as bit commitment. A bit commitment protocol traditionally consists of two phases: In the commit phase, Bob \emph{commits} a bit to Alice\footnote{Usually it is Alice who commits a bit to Bob. We decided to swap Alice and Bob as it allows us to simplify the notation in the proof of our main result. Throughout the paper it is Bob who commits a bit to Alice.}, who receives some form of confirmation that a commitment has been made. In the open phase, Bob reveals the bit to Alice. Security means that Bob should not be able to reveal anything but the committed bit, but nevertheless Alice cannot gain any information about the bit before the open phase. While many two-party cryptographic primitives have been defined, oblivious transfer and bit commitment are undoubtedly among the most important ones because they form essential building blocks for more complex problems~\cite{kilian88}.

Ideally, we would like to have protocols for such primitives that guarantee security without relying on any subjective (e.g.\ that a safe is difficult to open) or computational (e.g.\ that factoring a product of two large primes is difficult) assumptions. Unfortunately, however, it turned out that this is impossible, even if we allow quantum communication between Alice and Bob~\cite{mayers97, lo97a, dariano07, winkler11}. Much work has thus been invested into determining what kind of assumptions allow us to obtain security.
Of particular interest to this work are thereby assumptions of a physical nature, leading to information-theoretic security. Classical examples of such assumptions are, for example, access to some very special forms of shared randomness supplied in advance~\cite{rivest99}, access to a noisy communication channel~\footnote{To be more specific what is needed is a channel with a guaranteed level of noise. It is important that the noise is truly random and cannot be influenced by either party.}~\cite{crepeau97, winter03} or a limited amount of memory~\cite{maurer90}. Similarly, it has been shown that security is possible if the attacker's quantum memory is bounded~\cite{damgaard08, damgaard07, schaffner10} or more generally noisy~\cite{wehner08, koenig09, berta11}.

Another assumption is that of \emph{non-communication}. More precisely, one imagines that each party is split up into multiple agents who cannot communicate with each other for at least some parts of the protocol. Intuitively, the use of non-communicating agents can evade the standard no-go argument because while all agents in total have enough information to cheat, no single agent can cheat on his own.

On one hand, such non-communicating models have received considerable attention in classical cryptography, where such agents are often referred to as servers~\cite{kerenidis04} or provers~\cite{simard07}. For example, Ben-Or et al.~\cite{ben-or88} considered a simple protocol for bit commitment that is secure against classical attacks\footnote{Throughout this paper we will use the word classical to mean not quantum.} as long as the committer (Bob) is split up into two agents, Bob and Brian, who are not allowed to communicate throughout the protocol. This protocol can also be modified to give security against quantum adversaries~\cite{simard07}. Similarly, many classical protocols for other tasks have been proposed under the assumption of non-communication, such as distributed oblivious transfer~\cite{naor00}, i.e.\ symmetric private information retrieval~\cite{gertner98, malkin00, kerenidis04}, or simple private information retrieval~\cite{gasarch04}. In all such protocols it was assumed that the agents of one party can never communicate during any point in the protocol, or thereafter.

On the other hand, physicists have considered so-called relativistic assumptions for cryptography~\cite{kent99, kent05, colbeck06, kent11, kent12a}. 
In essence, this takes the form of non-communicating models where the fact that a party's agents cannot communicate is justified by their physical separation and the finite speed of light. The key difference to classical non-communicating models is that in relativistic models the separation is generally only imposed during certain periods of the protocol, whereas classical models generally assume a separation, i.e.\ non-communication, for all times. For example, relativistic protocols may only demand a split into several non-communicating agents after the commit phase of a bit commitment protocol is over~\cite{kent11, kent12a}. Another assumption based on relativity is the notion of guaranteed message delivery times (see Appendix~\ref{app:gmdt}) or the assumption of an accelerated observer\footnote{The authors consider two inertial participants sharing a noiseless quantum channel in the presence of a uniformly accelerated eavesdropper. They show that any information the eavesdropper manages to acquire is inherently noisy which allows the two honest participants to communicate securely. It is well-known in cryptography that most cryptographic primitives can be implemented securely as long as an external source of guaranteed noise is present.}~\cite{bradler09}.

Here, we will consider the security of bit commitment protocols under the assumption that one (or both) parties Alice and Bob, are forced to be split into non-communicating agents. Motivated by the relativistic protocols of~\cite{kent11, kent12a}, we thereby do \emph{not} demand that the parties are split into non-communicating agents for all time, but merely during certain phases of the protocol. A bit commitment protocol can be naturally divided into: the commit phase, the wait phase, the open phase, and the verification phase (see Section~\ref{sec:bitcom}). We thereby introduce the explicit notion of the wait and verification phases, which are usually only implicitly defined, in order to precisely divide the overall interaction between Alice and Bob into time frames. Our first contribution is
\begin{itemize}
\item A classification of non-communicating models into subclasses which are characterised by the phases in which Alice or Bob is split into non-communicating agents. We find that we can reduce our considerations to two minimal models, namely the one in which Alice is split during the commit and wait phases ($\alpha$-split) and the one in which Bob is split during the wait and open phases ($\beta$-split). Either of these two models allows to evade the no-go theorem because the operations required for cheating are forbidden by the split.
\end{itemize}
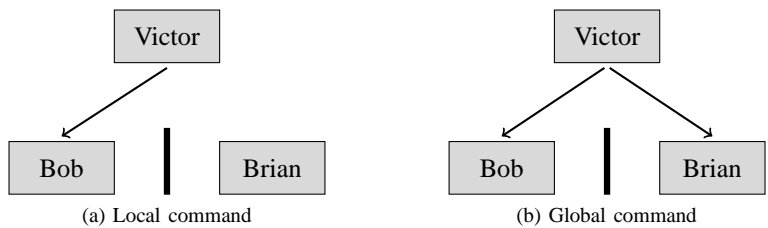
\begin{figure}[h]
\centering
\subfloat[Local command]
{
\begin{tikzpicture}[scale=0.7]
	\draw[fill=black] (2.95, 0) rectangle (3.05, 1.25);
	\draw[fill=gray!30] (0, 0) rectangle (2, 1);
	\draw[fill=gray!30] (4, 0) rectangle (6, 1);
	\draw[fill=gray!30] (2, 2.5) rectangle (4, 3.5);
	\node (Bob1) at (1, 0.5) {Bob};
	\node (Bob2) at (5, 0.5) {Brian};
	\node (Bob2) at (3, 3) {Victor};
	\draw [thick, ->] (3, 2.4) -- (1, 1.1);
\end{tikzpicture}
}
\hspace{36pt}
\subfloat[Global command]
{
\begin{tikzpicture}[scale=0.7]
	\draw[fill=black] (2.95, 0) rectangle (3.05, 1.25);
	\draw[fill=gray!30] (0, 0) rectangle (2, 1);
	\draw[fill=gray!30] (4, 0) rectangle (6, 1);
	\draw[fill=gray!30] (2, 2.5) rectangle (4, 3.5);
	\node (Bob1) at (1, 0.5) {Bob};
	\node (Bob2) at (5, 0.5) {Brian};
	\node (Bob2) at (3, 3) {Victor};
	\draw [thick, ->] (2.95, 2.4) -- (1, 1.1);
	\draw [thick, ->] (3.05, 2.4) -- (5, 1.1);
\end{tikzpicture}
}
\caption{If Bob is required to perform two separate openings it becomes important whether the command which bit he is supposed to unveil is transmitted to just one or both agents.}
\end{figure}
It turns out that in certain split models a new, subtle issue needs to be addressed. If a cheating Bob is split into two agents, Bob and Brian, during the open phase of the commitment, who decides which bit should be opened? In standard bit commitment protocols this question does not arise, as there is only one cheating party. Bob will simply announce to Alice that he wishes to unveil a particular bit, and try to provide a matching proof. However, in a model of several distinct agents, Bob and Brian could conceivably base the decision about which bit to unveil on some external input. For example, depending on the latest stockmarket news they both decide to open a $0$ or a $1$, even though they themselves cannot communicate. Intuitively, we would like a bit commitment scheme to be secure in the latter setting, analogous to the case of a single party which can of course also base its decision on external events. To capture this subtlety, we introduce an external verifier, Victor, who dictates which bit should be unveiled. We thereby speak of \emph{local} command if Victor only issues a command to one of the two agents, Bob. We speak of \emph{global} command if Victor issues a matching command to both Bob and Brian. Note that Victor should be thought of as an external verifier invoked solely to quantify Bob's cheating power and that he plays absolutely no role when both Alice and Bob are honest. The local and global command models will be defined in purely mathematical terms and the only reason to introduce Victor is to give these mathematical definitions some intuitive meaning. Note that a related concept has recently been introduced independently in~\cite{kent12b} under the name of the \emph{oracle input model}. In a model without separated agents, the local and global command models are equivalent but we will see that they differ in a relativistic setting. More precisely, our second contribution is to
\begin{itemize}
\item Introduce the distinction between local and global command in the models based on the $\beta$-split. We show that there is a simple classical protocol that is secure under the local command. However, we proceed to show that there exists \emph{no} classical protocol that is secure under global command in the class of $\beta$-split models.
\end{itemize}
The latter naturally leads to the question, whether there is a \emph{quantum} protocol that is secure even when Victor issues a global command. A quantum protocol that is likely to be secure under global command was given in~\cite{kent11}.
Another quantum $\beta$-split protocol was proposed by Kent~\cite{kent12a}, which has the very appealing feature that it can be implemented by the honest parties using only single qubit measurements in BB84~\cite{bb84} bases, without the use of any quantum memory. Yet, no explicit security bounds were provided in~\cite{kent12a}.
Our final contribution is to
\begin{itemize}
\item Provide a formal security proof and security bounds for the protocol proposed in~\cite{kent12a} in the global command model. 
\end{itemize}
We want to stress that a sketch of a security proof was given in~\cite{kent12a} already; however, we were unable to derive explicit security bounds from the arguments provided there. We thus devised an alternative proof, which allows us to find these parameters explicitly.

Our proof requires two ingredients: First, we make use of the fact that the two agents cannot communicate. Second, we employ an uncertainty relation in terms of min- and max-entropies~\cite{tomamichel11}. This relation was previously used to prove the security of quantum key distribution~\cite{tomamichellim11}, and our result illustrates its power to prove security of other cryptographic primitives.

{\bf Outline:} The paper is structured as follows. Section~\ref{sec:preliminaries} contains some basic definitions and technical tools essential for the proof. We also remind the reader what a bit commitment protocol is and what conditions it should satisfy. In Section~\ref{sec:relativistic-models} we introduce the concept of split models and, by examining the standard no-go argument, we find the minimal split requirements that might give us security and for these we state generalised security requirements. We also show how certain splits arise from special relativity if we require certain parts of the protocol to take place at space-like separated points. Section~\ref{sec:minimal-splits} presents simple protocols that achieve security in the minimal split models. Section~\ref{sec:transmitting-measurement-outcomes} is entirely dedicated to the bit commitment protocol proposed by Kent~\cite{kent12a}: first we describe the protocol and then we analyse its security to obtain explicit security bounds.

%
%
\section{Preliminaries}
\label{sec:preliminaries}
\subsection{Hamming distance}
Let $[n] = \{1, 2, \ldots, n\}$ and let $x$ be an $n$-bit string, $x \in \{0, 1\}^{n}$, and denote the $k$-th bit of $x$ by $x_{k}$. Define the Hamming distance between two strings $x, y \in \{0, 1\}^{n}$ to be the number of positions at which they differ
\begin{equation*}
\dham(x, y) := |\{k \in [n] : x_{k} \oplus y_{k} = 1\}|.
\end{equation*}
\subsection{Probability distributions}
\label{sec:prelim-probdist}
Let $X$ be a random variable taking values in $\cX$ and distributed according to $P_{X}$. The R{\'e}nyi entropy of order $\alpha \in \mathbb{R}_{+} \setminus \{0, 1, \infty\}$ is defined as~\cite{renyi61}
\begin{equation*}
\textnormal{H}_{\alpha}(X) := \frac{1}{1 - \alpha} \log \left(\sum_{x \in \cX} P_{X}(x)^\alpha \right).
\end{equation*}
The special cases $\alpha \in \{0, 1, \infty\}$ are defined as limits $\textnormal{H}_{\alpha}(X) = \lim_{\beta \to \alpha} \textnormal{H}_{\beta}(X)$. Note that $\textnormal{H}_{0}(X) = \log|\{x \in \cX : P_{X}(x) > 0\}|$ and that the R{\'e}nyi entropies exhibit monotonicity
\begin{equation*}
\textnormal{H}_{\alpha}(X) \geq \textnormal{H}_{\beta}(X) \iff \alpha \leq \beta.
\end{equation*}
For $\abs{\cX} = 2$ and $\alpha = 1$ we obtain the binary entropy
\begin{equation*}
h(q) := - q \log q - (1 - q) \log (1 - q).
\end{equation*}
Let $P_{XY | UV}$ be a joint conditional probability distribution. $P_{XY | UV}$ satisfies no-signalling if for all $u \in \cU, x \in \cX$ the value of the sum
\begin{equation*}
\sum_{y \in \cY} P_{XY|UV}(X = x, Y = y | U = u, V = v)
\end{equation*}
does not depend on a particular choice of $v \in \cV$.
\subsection{Quantum notation}
Let $\rho$ be a quantum state on a Hilbert space $\cH$, i.e. a positive semi-definite operator with $\tr \rho = 1$ acting on $\cH$. Let $\cS(\cH)$ be the set of all states on $\cH$. We say that $\rho_{XA}$ is a classical-quantum (cq) state if it can be written in the form
\begin{equation*}
\rho_{XA} = \sum_{x \in \cX} P_{X}(x) \ketbra{x}_{X} \otimes \rho_{x},
\end{equation*}
where $P_{X}$ is a probability distribution and $\rho_{x} \in \cS(\cH_{A})$. Then, we define the probability of guessing $X$ given access to the quantum system $A$ as
\begin{equation*}
p_{\textnormal{guess}}(X | A) := \max_{\{M_{x}\}} \sum_{x \in \cX} P_{X}(x) \tr (M_{x} \rho_{x}),
\end{equation*}
where the maximisation is taken over all positive operator-valued measurements (POVMs) on $\mathcal{H}_{A}$. The \emph{min-entropy} of $X$ is defined as $\hmin(X) := \textnormal{H}_{\infty}(X)$. The min-entropy of $X$ conditioned on $A$ is defined as
\begin{equation*}
\hmin(X | A) := - \log p_{\textnormal{guess}}(X | A).
\end{equation*}
We say that $\rho_{XY}$ is a classical-classical (cc) state if it can be written in the form
\begin{equation*}
\rho_{XY} = \sum_{x \in \cX, y \in \cY} P_{XY}(x, y) \ketbra{x}_{X} \otimes \ketbra{y}_{Y}.
\end{equation*}
The \emph{max-entropy} of $X$ is defined as $\hmax(X) := \textnormal{H}_{\frac{1}{2}}(X)$. The max-entropy of $X$ conditioned on $Y$ is defined as
\begin{equation*}
\hmax(X | Y) := \log \sum_{y \in \cY} \Pr[Y = y] \cdot 2^{\hmax(X | Y = y)}.
\end{equation*}
%
%
%
%
\subsection{Uncertainty relation}
Let $\rho_{ABC}$ be any tri-partite state and let $\{M_{z}\}_{z \in \cZ}$ and $\{N_{x}\}_{x \in \cX}$ be two POVMs on the $A$ subsystem whose measurement results are represented by classical random variables $Z$ and $X$. The following cq-states arise from performing the measurements mentioned above\footnote{To simplify the notation we will omit all the subsystems on which the projector equals identity. Hence, in our shorthand notation $M_{z} \rho_{ABC}$ stands for $(M_{z} \otimes \mathbb{I}_{BC}) \rho_{ABC}$.} :
\begin{align*}
\rho_{ZB} &:= \sum_{z \in \cZ} \ketbra{z}_{Z} \otimes \tr_{AC} (M_{z} \rho_{ABC}) \quad \mbox{and}\\
\rho_{XC} &:= \sum_{x \in \cX} \ketbra{x}_{X} \otimes \tr_{AB} (N_{x} \rho_{ABC}).
\end{align*}
\begin{thm}
\label{thm:uncertainty}
\cite{tomamichel11} For any tri-partite state $\rho_{ABC}$ the following uncertainty relation holds
\begin{equation}
\label{eq:uncertainty}
\hmax(Z | B) + \hmin(X | C) \geq \log \frac{1}{c},
\end{equation}
where the entropies are evaluated for $\rho_{ZB}$ and $\rho_{XC}$, respectively, and $c := \max_{z, x} \norm{\sqrt{M_{z}} \sqrt{N_{x}}}_{\infty}^{2}$.
\end{thm}
\subsection{Bit commitment}
\label{sec:bitcom}
Bit commitment is a primitive that allows Bob to commit a bit $b$ to Alice in a way that is both binding (Bob cannot later convince Alice that he actually committed to $1 - b$) and hiding (Alice cannot figure out what $b$ is before Bob decides to unveil it). In this section we discuss how to describe a bit commitment protocol\footnote{Note that we do not consider the most general class of protocols as we assume that the open phase involves one-way communication from Bob to Alice only.} and how to formalise the desired security requirements.

Any action taken by Alice or Bob can be described by a completely positive, trace-preserving (CPTP) map and the entire protocol can be defined by specifying these maps. In this paper we will denote maps performed by Alice and Bob by  $\Lambda$ and $\Phi$, respectively. The subscript $X \to Y$ means that the map acts on (reads and/or modifies) the existing register $X$ and creates a new register $Y$. Moreover, identity is assumed on any subsystems not explicitly mentioned within the map: $\Lambda_{X \to Y} (\rho_{XYZ})$ stands for $(\Lambda_{X \to Y} \otimes \id_{Z}) (\rho_{XYZ})$.
%
%
%
%
%
%
%
%
%

The usual description of a bit commitment protocol divides it into two phases: commit and open. However, as our scenarios rely on timing and communication constraints, it is useful to be more explicit about the structure of the protocol. We divide the protocol into four phases: \emph{commit}, \emph{wait}, \emph{open} and \emph{verify}.
The commit and open phases are the essence of the protocol: they are the only phases during which Alice and Bob interact. The wait phase acts merely as a separator (this is when the commitment is valid), while in the verify phase Alice uses the information collected in the previous phases to verify the commitment and decide whether to accept or reject it.

Let $\rho_{ABC}$ be the state that Alice and Bob share at the end of the commit phase if they are both honest.\footnote{Any private or shared randomness is included in the description of the state, hence, given a protocol we can extract a unique $\rho_{ABC}$.} The subsystems $A$ and $B$ are controlled by Alice and Bob, respectively, while subsystem $C$ is a classical register in Bob's posession indicating which bit Bob has (honestly) committed to. Let $\Phi_{BC \to P}^{\textnormal{open}}$ be the quantum operation that Bob applies in the open phase and it should be thought of as extracting a proof of his commitment from the subsystems $B$ and $C$ and storing it in the (possibly quantum) subsystem $P$\footnote{The honest opening map will simply read the value of the classical register $C$, hence, its state will not be affected.}
\begin{equation*}
\rho_{ABPC} = \Phi_{BC \to P}^{\textnormal{open}} (\rho_{ABC}).
\end{equation*}
In the last step of the open phase Bob passes the subsystems $P$ and $C$ to Alice. Note that as $C$ is a classical register Alice is automatically assumed to read it and, hence, she finds out what Bob claims to have commited to. Let $\Lambda_{APC \to F}^{\textnormal{verify}}$ be the quantum operation that Alice applies in the verify phase, which creates a classical binary register (flag), $F$, indicating whether the commitment is accepted or rejected
\begin{equation*}
\rho_{ABPCF} = \Lambda_{APC \to F}^{\textnormal{verify}} (\rho_{ABPC}),
\end{equation*}
and let us denote a (classical) basis of the subsystem $F$ by $\{\ket{\textnormal{accept}},\ket{\textnormal{reject}}\}$. Describing the honest protocol suffices to define correctness.
\begin{df}
A bit commitment protocol is \emph{perfectly correct} if $\rho_{ABC}$ satisfies
\begin{equation*}
\bramatket{\textnormal{accept}}{\tr_{ABPC} \Lambda_{APC \to F}^{\textnormal{verify}} ( \Phi_{BC \to P}^{\textnormal{open}} (\rho_{ABC}))}{\textnormal{accept}} = 1.
\end{equation*}
\end{df}
If one of the parties is dishonest and does not follow the protocol then the state shared between Alice and Bob is no longer well-defined. We will use $\sigma$ to denote such a dishonest state\footnote{We make no assumptions on what the dishonest party stores in their part of the state. In particular it might contain some ancillary systems to be used later.} to distinguish them from the honest states denoted by $\rho$. Security guarantee for honest Bob states that Alice finds it difficult to guess the value of his commitment before the open phase. If Alice is dishonest and does not follow the protocol then the state shared at the end of the commit phase, $\sigma_{ABC}$, does not necessarily equal $\rho_{ABC}$. However, it is important to note that the classical register $C$ is still well-defined since Bob is honest. Let $\cK_{A}$ be the set of all tri-partite states that Alice might enforce at the end of the commit phase. Informally, a bit commitment is $\delta$-hiding if for any cheating strategy the probability that Alice guesses the committed bit correctly before the open phase is upperbounded by $\frac{1}{2} + \delta$.
%
%
%
\begin{df}
\label{df:hiding}
A bit commitment protocol is \emph{$\delta$-hiding} if all $\sigma_{ABC} \in \cK_{A}$ satisfy
\begin{equation*}
p_{\textnormal{guess}}(C|A) \leq \frac{1}{2} + \delta.
\end{equation*}
\end{df}
Similarly, if Bob is dishonest then different states may be reached at the end of the commit phase and let $\cK_{B}$ be the set of all states that he might enforce at the end of the commit phase. Note that the classical register $C$ is no longer well-defined so we will simply talk about bi-partite states $\sigma_{AB} \in \cK_{B}$. In order to cheat successfully Bob must be able to produce valid proofs for both values of $C$, which implies that there are two distinct dishonest opening maps: Bob applies $\Phi_{B \to PC}^{\textnormal{cheat}, 0}$ if he chooses to open $0$ and $\Phi_{B \to PC}^{\textnormal{cheat}, 1}$ if he chooses to open $1$. The cheating map $\Phi_{B \to PC}^{\textnormal{cheat}, b}$ extracts the proof of having committed to $b$ from the subsystem B, stores it in the subsystem $P$ and stores $b$ in the newly-created register $C$
\begin{equation*}
\sigma_{ABP}^{b} \otimes \ketbra{b}_{C}  = \Phi_{B \to PC}^{\textnormal{cheat}, b} (\sigma_{AB}).
\end{equation*}
In the last step Bob gives $P$ and $C$ to Alice, who verifies the commitment using the honest map. Let $p_{b}$ be the probability that Alice accepts Bob's unveiling of $b$
\begin{equation}
\label{eq:pb}
p_{b} = \bramatket{\textnormal{accept}}{\tr_{ABPC} \Lambda_{APC \to F}^{\textnormal{verify}} (\Phi_{B \to PC}^{\textnormal{cheat}, b} (\sigma_{AB}))}{\textnormal{accept}}.
\end{equation}
The security conditions on $p_{0}$ and $p_{1}$ depend on whether we are in the classical or quantum framework. Classically, we require that at the end of the commit phase at least one of $\{p_{0}, p_{1}\}$ is small. However, this requirement turns out to be too strong in the quantum world as explained in~\cite{dumais00} and a weaker security condition is proposed in the same paper.
\begin{df}
\label{df:weaklybnd}
A bit commitment protocol is \emph{$\varepsilon$-weakly binding} if for all $\sigma_{AB} \in \cK_{B}$ and for all cheating maps $\{\Phi_{B \to PC}^{\textnormal{cheat},b}\}_{b \in \{0, 1\}}$ we have $p_{0} + p_{1} \leq 1 + \varepsilon$.
\end{df}
Unfortunately, this definition does not give us composability (see Appendix~\ref{app:counterexample} for a counter-example). On the other hand the usual composable definition used for quantum protocols introduced in~\cite{damgaard07} turns out to be too stringent for the scenarios considered in this paper (see Appendix~\ref{sec:strongerdefimpossible} for details). Hence, throughout the paper we will stick to the weaker, non-composable definition.
\section{Relativistic models}
\label{sec:relativistic-models}
Before considering relativistic models let us briefly examine the original no-go argument (for the full version please refer to~\cite{mayers97, lo97a}) to see how it might be circumvented by imposing certain communication constraints.
\subsection{The original no-go argument and the split models}
\label{sec:no-go-and-split}
First note that we can restrict ourselves to protocols in which the state shared between Alice and Bob is pure at all times.\footnote{We assume that Alice and Bob start in a pure state and then all the actions can be performed coherently.} Let $\ket{\phi_{AB}^{b}}$ be the state at the end of the commit phase if Bob has decided to commit to $b$. We require that Alice should not be able to distinguish the two cases just by looking at her subsystem which implies that $\rho_{A}^{0} = \rho_{A}^{1}$, where $\rho_{A}^{b} = \tr_{B} \ketbra{\phi_{AB}^{b}}$. By Uhlmann's theorem~\cite{uhlmann85} there exists a unitary $U_{B}$ acting on the subsystem $B$ alone such that $U_{B} \ket{\phi_{AB}^{0}} = \ket{\phi_{AB}^{1}}$. Hence, if the states corresponding to both commitments are the same on Alice's side then Bob can cheat perfectly. This argument can be extended to the case in which $\rho_{A}^{0}$ and $\rho_{A}^{1}$ are close in trace distance (which means that they are difficult to distinguish) and then one can show that Bob can still cheat with high probability (for the exact trade-off based on this idea refer to~\cite{spekkens01}; for the optimal bounds on quantum bit commitment see~\cite{chailloux11}).

What is a split model? Informally, a split model is a model in which at least one party is required to delegate multiple agents to perform certain parts of the protocol in a non-communicating fashion. In this paper we only consider models in which we require a party to delegate at most $2$ agents. The basic rule of two-party cryptography is that there are no third parties: the world is split between Alice and Bob only, anything that does not belong to Alice is fully controlled by Bob. Now suppose that the split model requires that there are two agents of Bob (Bob and Brian). It is still true that Bob and Brian \emph{together} control everything that does not belong to Alice. However, the class of operations they can perform in a non-communicating fashion is now restricted, which might give us security. It is clear that the only way to achieve security is to split Alice during the period for which security for Bob should hold or \emph{vice versa}. Therefore, we arrive at two relevant splits.
\begin{itemize}
\item $\alpha$-split : Alice is split during the commit and wait phases.
\item $\beta$-split : Bob is split during the wait and open phases.
\end{itemize}
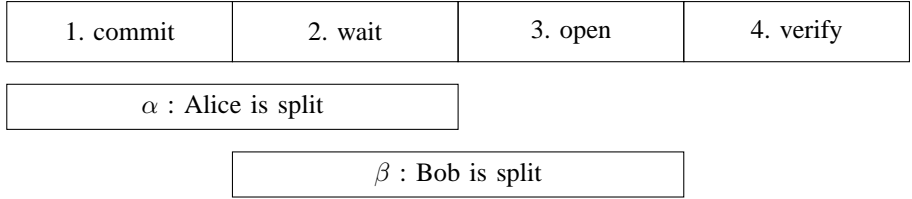
\begin{figure}[h]
\centering
\begin{tikzpicture}[scale=1]
\draw (0, 0) rectangle (3, 0.8);
\node (commit) at (1.5, 0.4) {1. commit};
\draw (3, 0) rectangle (6, 0.8);
\node (wait) at (4.5, 0.4) {2. wait};
\draw (6, 0) rectangle (9, 0.8);
\node (open) at (7.5, 0.4) {3. open};
\draw (9, 0) rectangle (12, 0.8);
\node (verify) at (10.5, 0.4) {4. verify};
\draw (0, -0.3) rectangle (6, -0.9);
\node (alpha) at (3, -0.6) {$\alpha$ : Alice is split};
\draw (3, -1.2) rectangle (9, -1.8);
\node (beta) at (6, -1.5) {$\beta$ : Bob is split};
\end{tikzpicture}
\caption{The two relevant types of separations: $\alpha$ and $\beta$.}
\label{fig:alphabeta}
\end{figure}
The standard no-go does not apply to the $\alpha$-split model because while $\rho_{A}^{0}$ might be globally fully distinguishable from $\rho_{A}^{1}$ they might locally look the same for both Alice and Amy (her agent). The $\beta$-split evades the no-go because the global unitary $U_{B}$ might be impossible to perform by Bob and Brian without communication. Note that whenever we say that a party is split during two (or more) consecutive phases of the protocol we mean one long split throughout the whole period rather than a sequence of short ones (the agents are \emph{not} allowed to get together in between).

We treat the splits as a resource. Hence, we are interested in the minimal splits that give security and we will show that $\alpha$ and $\beta$ are such minimal splits. What about models that impose strictly more restrictions than those? On one hand any protocol secure in the minimal split will remain secure in the more split model, we only need to ensure it is still feasible. E.g.~the protocol from~\cite{kent12a} was originally proposed in the model in which \emph{both} Alice and Bob are split during the wait and open phases, while our analysis applies to the $\beta$-split model (strictly less split). Therefore, our proof automatically extends to the original setting. On the other hand, imposing more split might allow for new, simpler protocols. E.g.~for the case of Bob being split at all times there exists a number of protocols~\cite{ben-or88, simard07, kent99, kent05}.

The number of possible split models is rather large and examining all of them case-by-case is unlikely to give any valuable insight. Hence, in this paper we only focus on the minimal splits: $\alpha$ and $\beta$.
%
%
It is clear that a split imposed on Alice will only affect her cheating power (not Bob's) and it is only the security guarantee for honest Bob that needs to be generalised. In the $\alpha$-split Bob commits to a bit by talking to Alice and Amy (subsystems $A$ and $A'$, respectively) and a natural generalisation of the hiding condition is to require that $\emph{neither}$ of them acquires significant knowledge about the value of $C$. In analogy to the non-split case let $\cK_{AA'}$ be the set of states that dishonest Alice and Amy can enforce at the end of the commit phase. Then the split counterpart of Definition~\ref{df:hiding} can be written as follows.
\begin{df}
An $\alpha$-split bit commitment protocol is \emph{$\delta$-hiding} if all $\sigma_{AA'BC} \in \cK_{AA'}$ satisfy
\begin{equation*}
p_{\textnormal{guess}}(C|X) \leq \frac{1}{2} + \delta \nbox{for}{10} X = \{A, A'\}.
\end{equation*}
\end{df}
Similarly, in the $\beta$-split let $\cK_{BB'}$ be the set of states that dishonest Bob and Brian can enforce at the end of the commit phase. In the introduction we mentioned the concept of an external verifier Victor who challenges Bob to open a particular bit and this is how we quantify Bob's cheating power. In the case of Bob and Brian performing two openings separately we need to specify whether Victor only tells Bob what to unveil or both Bob and Brian receive the message. We call these two scenarios the \emph{local} and \emph{global} command models, respectively. The first variant corresponds to the situation in which Bob makes the decision while Brian intends to behave consistently. If $b$ is the bit that Bob intends to unveil then the cheating maps in the local command model take the form
\begin{equation*}
\Phi_{BB' \to PP'CC'}^{\textnormal{cheat}, \textnormal{local}, b} = \Phi_{B \to PC}^{\textnormal{cheat}, b} \otimes \Phi_{B' \to P'C'}^{\textnormal{cheat}},
\end{equation*}
i.e.~Bob's actions depend on $b$ but Brian's behaviour is independent of it.

The natural motivation for the second scenario is a situation in which the agents are not allowed to communicate with each other but they might receive information from the outside world, hence, they both know $b$. The cheating maps in the global command model take the form
\begin{equation*}
\Phi_{BB' \to PP'CC'}^{\textnormal{cheat}, \textnormal{global}, b} = \Phi_{B \to PC}^{\textnormal{cheat},b} \otimes \Phi_{B' \to P'C'}^{\textnormal{cheat}, b},
\end{equation*}
i.e.~both opening maps depend on the value of $b$. Using the definition of $p_{b}$, the probability of successfully opening $b$, introduced in~\eqref{eq:pb} we can state the security condition in the $\beta$-split model.
\begin{df}
\label{df:weaklybnd}
A $\beta$-split bit commitment protocol is \emph{$\varepsilon$-weakly binding} in the local (global) command model if for all $\sigma_{ABB'} \in \cK_{BB'}$ and all the cheating maps allowed in the local (global) command model we have $p_{0} + p_{1} \leq 1 + \varepsilon$.
\end{df}
The two variations of the $\beta$-split model turn out to be rather different from the security point of view: there exist simple classical protocols secure in the local command model, while no classical protocol can be secure in the global command model (for details please refer to Section~\ref{sec:betasep}). Hence, to satisfy this stronger security requirement one needs to resort to quantum protocols and we investigate one of them in Section~\ref{sec:transmitting-measurement-outcomes}.
\begin{figure}[h!]
\centering
\begin{tikzpicture}[scale=0.7]
	\coordinate (y) at (0, 9);
	\coordinate (x) at (10, 0);
	\coordinate (origin) at (0, 0);
	\coordinate (p) at (5, 1);
	\coordinate (q) at (2, 4);
	\coordinate (r) at (8, 4);
	\coordinate (t) at (5, 7);
	\coordinate (t0) at (0, 1);
	\coordinate (t1) at (0, 4);
	\coordinate (t2) at (0, 7);
	\coordinate (x1) at (2, 0);
	\coordinate (x0) at (5, 0);
	\coordinate (x2) at (8, 0);
	\path [fill=gray!15] (q) -- (p) -- (4, 0) -- (origin) -- (0, 2);
	\path [fill=gray!15] (q) -- (t) -- (4, 8) -- (0, 8) -- (0, 5.5);
	\path [fill=gray!15] (r) -- (p) -- (6, 0) -- (10, 0) -- (10, 2);
	\path [fill=gray!15] (r) -- (t) -- (6, 8) -- (10, 8) -- (10, 5.5);
	\path [fill=gray!40] (4, 0) -- (p) -- (6, 0);
	\path [fill=gray!40] (4, 8) -- (t) -- (6, 8);
	\draw[<->] (y) node[left] {$t$} -- (origin) --  (x) node[below] {$x$};
	\filldraw[black] (p) circle (2pt) node[right] {$P$};
	\filldraw[black] (q) circle (2pt) node[above] {$Q$};
	\filldraw[black] (r) circle (2pt) node[above] {$R$};
	\filldraw[black] (t) circle (2pt) node[above] {$T$};
	\draw[dashed] (t0) node[left] {$0$} -- (p);
	\draw[dashed] (t1) node[left] {$1$} -- (r);
	\draw[dashed] (t2) node[left] {$2$} -- (t);
	\draw[dashed] (x0) node[below] {$0$} -- (t);
	\draw[dashed] (x1) node[below] {$-1$} -- (q);
	\draw[dashed] (x2) node[below] {$1$} -- (r);
\end{tikzpicture}
\caption{Light gray regions represent the light cones of $Q$ and $R$, while dark gray corresponds to the common past or future. $P$ is the latest point of the common past, while $T$ is the earliest point of the common future.}
\label{fig:cones}
\end{figure}
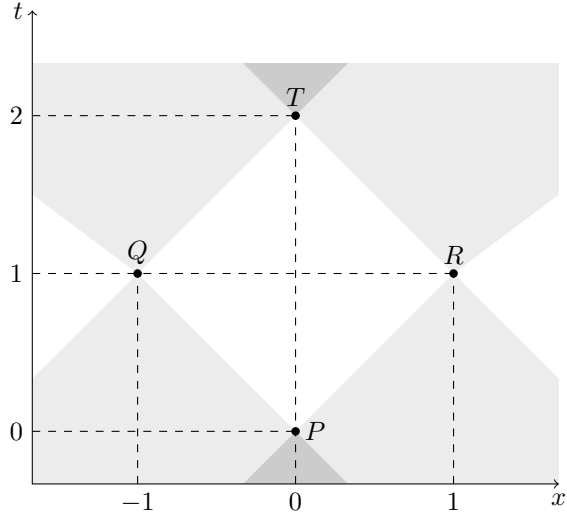
\subsection{Relativistic motivation}
\label{sec:relmotivation}
Special relativity states that information cannot travel faster than the speed of light. Hence, if we are guaranteed that sites $X$ and $Y$ are at some well-defined distance we can calculate the minimum time it takes for a message to travel from $X$ to $Y$ (or \emph{vice versa}). This motivates \emph{guaranteed message delivery time} models, in which transmitting messages between certain parties takes a finite amount of time. To the best of our knowledge, these were the first models in which relativistic bit commitment was proposed~\cite{kent99, kent05} (please refer to Appendix~\ref{app:gmdt} for a brief summary of what is known about these models). Special relativity can also motivate certain split models as explained below.
%

We consider the model proposed by Kent~\cite{kent11, kent12a}. Take the speed of light to be $1$, let $(x, t)$ be the coordinates for Minkowski space and define the following three points : $P = (0, 0)$, $Q = (- 1, 1)$, $R = (1, 1)$. It is clear that $P$ is the latest point that belongs to the common past of $Q$ and $R$ (Fig.~\ref{fig:cones}). Hence, no signal emitted after $t = 0$ (regardless of where it was emitted from) can reach both $Q$ and $R$. Kent's bit commitment protocols take advantage of this scenario by assuming that each party has an agent at $P$, $Q$ and $R$ and they are allowed to send information at the speed of light. The commit phase happens at $P$ while the open phase happens at $Q$ and $R$. The resulting communication constraints are illustrated in Fig.~\ref{fig:kentsplit}. It is clear that the communication constraints following from this configuration in space-time are strictly stronger than those of the $\beta$-split. This serves as a proof of principle that at least certain split models can be physically realised by requiring different parts of the protocol to take place at different, space-like separated points.
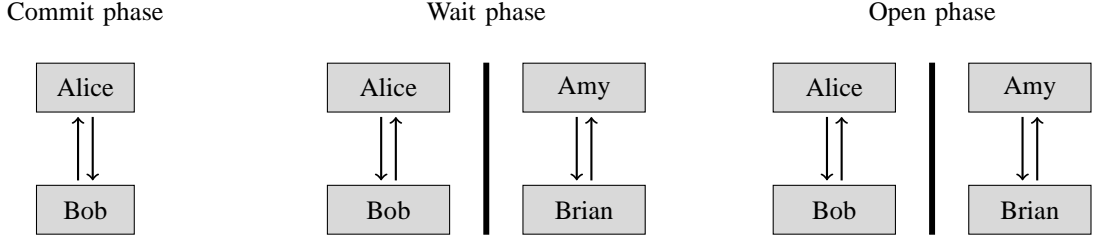
\splits{1}{4}{4}{0.65}{Effective communication constraints imposed by Kent's model~\cite{kent11, kent12a}.}{fig:kentsplit}
\section{Bit commitment protocols for the minimal splits}
\label{sec:minimal-splits}
In Section \ref{sec:no-go-and-split} we argued that either $\alpha$ or $\beta$-split needs to be imposed for security to be possible. In this section we give explicit examples of protocols which are secure in each of the two cases.
\subsection{Protocols based on $\alpha$-split}
\label{sec:alphasep}
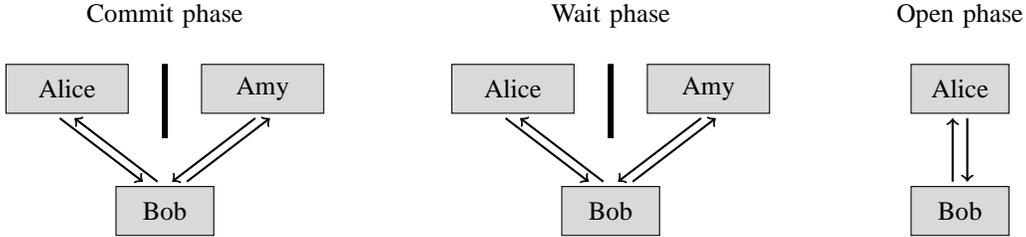
\splits{2}{2}{1}{0.65}{The $\alpha$-split model: Alice is required to be split during the commit and wait phases.}{}
The $\alpha$-split allows for a simple bit commitment protocol based on secret sharing. Such protocols will have the feature that once the commit phase is over, the combined systems of Alice and Amy determine the committed bit and the commitment only lasts as long as the separation is maintained. This is similar to the distributed oblivious transfer scenarios \cite{naor00} in which security disappears as soon as the agents are allowed to communicate.
\begin{prot}{Bit commitment from secret sharing}
\label{prot:secshar}
\begin{enumerate}
\item (commit) Bob commits to $b \in \{0, 1\}$ by generating a random bit $r$ and sending $b \oplus r$ to Alice and $r$ to Amy.
\item (open) Alice and Amy calculate $b = (b \oplus r) \oplus r$.
\end{enumerate}
\end{prot}
Security against classical adversaries follows directly from the properties of secret-sharing. It is also secure against quantum adversaries (see Appendix~\ref{app:secshasec} for details). As there exists a classical protocol that is perfectly secure (even against quantum adversaries) in this scenario quantum mechanics gives us no advantage for the purpose of bit commitment.
\subsection{Protocols based on $\beta$-split}
\label{sec:betasep}
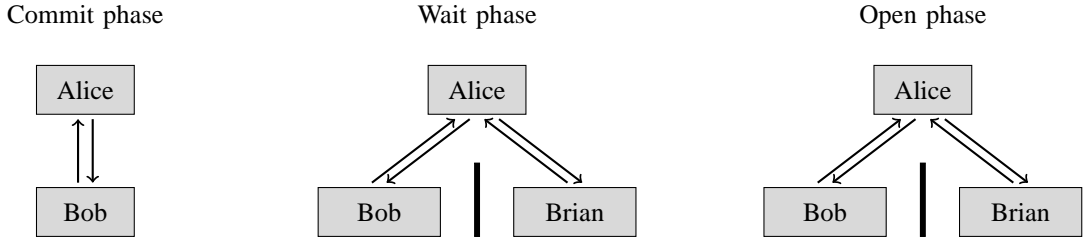
\splits{1}{3}{3}{0.65}{The $\beta$-split model: Bob is required to be split during the wait and open phases.}{}
In contrast to the $\alpha$-split case commitments based on the $\beta$-split can be made permanent\,---\,Bob and Brian can always refuse to participate in the open phase and Alice will learn nothing about their commitment. As discussed in Section~\ref{sec:no-go-and-split} we need to distinguish between the local and global command models.
\subsubsection{Security in the local command model}
\label{sec:bobsplitopen}
It turns out that in the $\beta$-split model under the local command there exists a simple classical protocol that achieves security.
\begin{prot}{Bit commitment in the local command model}
\label{prot:trivlocal}
\begin{enumerate}
\item (commit) Bob chooses a bit $b$ and shares it with Brian.
\item (open) Bob and Brian independently send to Alice a bit they claim to have committed to (denote these bits by $x$ and $y$, respectively).
\item (verify) Alice accepts the commitment of $b$ if $b = x = y$, else she rejects.
\end{enumerate}
\end{prot}
It is easy to convince ourselves that the protocol is secure (according to the weakly binding definition). The problem that Bob and Brian face is to correlate the bits they are trying to unveil. In order to do that they either have to agree on the bit in advance (which corresponds to an honest commitment) or they would have to violate no-signalling. For a more detailed security analysis we refer to Appendix~\ref{app:loccomsec} (see also the independent discussion of this and related points in~\cite{kent12b}).

\subsubsection{Security in the global command model}
We have seen that in the local command model there exists a very simple bit commitment protocol that achieves security. Unfortunately, as soon as we switch to the global command the protocol becomes insecure\,---\,Bob and Brian can cheat perfectly. Let us consider what is and what is not possible in the $\beta$-split model under the global command.
\paragraph{Classically}
Classically, it is not possible to achieve security in the $\beta$-split model under the global command and the informal argument goes as follows. As the protocol needs to be correct Bob and Brian must be able to honestly commit to either bit, i.e.~they must be able to agree on unveiling strategies\footnote{Bob and Brian agree on unveiling strategies during the commit phase, which they are allowed in the $\beta$-split model. This argument might not apply in the case of stronger splits (e.g.~Bob and Brian split at all times).} that will make Alice accept either bit even without any further communication between Bob and Brian. Since the protocol is hiding the interaction during the commit phase cannot give away any information about the committed bit and, therefore, both strategies remain valid until the beginning of the open phase. Hence, whichever bit Bob and Brian are told to unveil they can always succeed.
\paragraph{Quantum mechanically}
The informal argument presented above does not apply in the quantum world due to the no-cloning principle. The opening strategy may rely on some quantum system that is available to Bob right before the split\,---\,but cannot be shared with Brian without loss. The first protocols in the $\beta$-split model were proposed by Kent~\cite{kent11, kent12a} and Section~\ref{sec:transmitting-measurement-outcomes} focuses on one of them.
\section{Bit commitment by transmitting measurement outcomes}
\label{sec:transmitting-measurement-outcomes}
We introduce a variant of the bit commitment protocol by Kent~\cite{kent12a} and then present a security proof that leads to explicit security bounds. 
\subsection{The protocol}
The original protocol presented in~\cite{kent12a} uses BB84 states. However, for the purpose of the proof we analyse its purified analogue (which is equivalent from the security point of view). Denote the computational basis by $\cB_{0} = \{\ket{0}, \ket{1}\}$ and the Hadamard basis by $\cB_{1} = \{\ket{+}, \ket{-}\}$.

Note also that the original scenario described by Kent makes strictly more assumptions (because it requires both parties to be split rather than just one). However, we will see that whether Alice is split or not does not affect the security. Hence, the security proof for the $\beta$-split model presented here automatically applies to the setup originally proposed by Kent.

\begin{prot}{Bit commitment by transmitting measurement outcomes}
\label{prot:trans}
\begin{enumerate}
\item Alice creates $2n$ EPR pairs and sends one half of each pair to Bob.
\item (commit) Bob commits to a bit $b$ by measuring every qubit he receives in $\cB_{b}$. Denote the outcomes by $T$ (a classical bit string of length $2n$).
\item (end of commit) Bob splits up into two agents: Bob and Brian. Each of them holds a copy of $T$. {\bf No more communication is allowed between Bob and Brian until the end of the protocol.}
\item (open) Bob opens the commitment by sending $b$ and $T$ to Alice. Brian does the same.
\item Alice picks a random subset $\cZ \subset [2n]$ of size $n$ and let $\cX := [2n] \setminus \cZ$. She measures the qubits from $\cZ$ in the computational basis and the qubits from $\cX$ in the Hadamard basis. Denote her measurement outcomes by $S$ (a classical bit string of length 2n).
\item (verify) Alice performs three checks :
\begin{itemize}
\item Alice checks whether the values of $b$ submitted by Bob and Brian are the same.
\item Alice checks whether the strings submitted by Bob and Brian are the same.
\item Alice checks whether the strings submitted are consistent with $S$ (consistency check on qubits she measured in $\cB_{b}$ only).
\end{itemize}
If all three checks pass then the opening is accepted.
\end{enumerate}
\end{prot}
As mentioned in Section~\ref{sec:bitcom} a secure bit commitment protocol should satisfy three conditions. If Bob is honest he will choose a bit $b$, perform the correct measurement to obtain the (classical) string $T$. After the split Bob and Brian will both possess identical copies of $b$ and $T$, which they send to Alice during the open phase. Hence, the first two checks clearly go through. The third check goes through because honest Alice prepared perfect EPR pairs, measured them to obtain string $S$ and so strings $S$ and $T$ must be perfectly correlated on the qubits measured in the same basis. Hence, the protocol is perfectly correct. Security for honest Bob is also easy to see. Alice does not receive any information before the open phase, hence, she cannot learn anything about Bob's commitment by no-signalling and the protocol is $\delta$-hiding for $\delta = 0$. Therefore, we only analyse security for honest Alice, i.e.~show the following result:
\begin{thm}
\label{thm:mainresult}
Protocol~\ref{prot:trans} in the $\beta$-split model under the global command is $\varepsilon$-weakly binding, where
\begin{equation*}
\varepsilon = \inf_{\delta \in (0, \frac{1}{2}) } 2^{1 - n(1 - h(\delta))} + 2 \exp \left(- \frac{1}{2} n \delta^{2} \right),
\end{equation*}
where $h(\cdot)$ is the binary entropy function as defined in Section~\ref{sec:prelim-probdist}.
\end{thm}
Note that not only does $\varepsilon$ vanish in the limit $n \to \infty$ but also the rate of decay is exponential in $n$ ($n$ is the number of rounds played, hence, the resources necessary to execute the protocol grow linearly in $n$). The fact that $\varepsilon$ decays exponentially would be a great advantage if the protocol were to be implemented experimentally and shows that the protocol might be of practical interest.
\subsection{Security for honest Alice}
\label{sec:secalice}
%
%
%
%
%
%
%
%
%
\subsubsection{Notation}
Let us denote the state of the system at the end of the commit phase by $\sigma_{AB B'}$, where subsystems $A$, $B$ and $B'$ belong to Alice, Bob and Brian, respectively. Alice is honest so we know the exact state of her subsystem\,---\,it contains $2n$ qubits, which have already been partitioned into sets $\cZ$ and $\cX$. This justifies a natural partition of the subsystem $A$ into subsystems $A_{\cZ}$ and $A_{\cX}$, each containing exactly $n$ qubits. Let quantum operation $\Lambda_{G}^{b}$ for $G \in \{A_{\cZ}, A_{\cX}\}$, $b \in \{0, 1\}$ correspond to measuring all qubits from the subsystem $G$ in the basis $\cB_{b}$. The relevant projectors can be formally defined as
\begin{equation}
P_{G}^{b, s} := [H^{\otimes n}]^{b} \ketbra{s}_{G} [H^{\otimes n}]^{b},
\end{equation}
where $s \in \{0, 1\}^{n}$. Denote the environment by $E$ and the subsystem used to store the measurement outcomes by $F$. Then $\Lambda_{G}^{b}$ is defined as
\begin{equation*}
\label{eq:meas1}
\rho_{FE} := \Lambda_{G}^{b} (\rho_{GE}) = \sum_{s} \ketbra{s}_{F} \otimes \tr_{G} \big( P_{G}^{b, s} \rho_{GE} \big).
\end{equation*}
The three relevant measurements are $\Lambda_{A_{\cZ}}^{0}, \Lambda_{A_{\cX}}^{1}, \Lambda_{A_{\cX}}^{0}$\,---\,the first two are actually performed in the honest protocol, while the third one is a \emph{virtual} measurement, required for the proof only. Bob and Brian are expected to extract a string from their respective quantum systems. Let us simplify the notation introduced in Section~\ref{sec:bitcom} and denote Bob's map intending to open $b$ and producing string $T$ as the output by $\Phi_{B}^{b}$. Similarly, for Brian denote the map intending to open $b'$ by $\Phi_{B'}^{b'}$ and the output string by $T'$. Observe that $\Phi_{B}^{b}$ ($\Phi_{B'}^{b'}$) is restricted to operate on the subsystem $B$ ($B'$) only. The string $T$ corresponds to measuring all $2n$ qubits. Once Alice has chosen the partition into $\cZ$ and $\cX$ we can naturally split it into two substrings $T = \{T_{\cZ}, T_{\cX}\}$, which correspond to the outcomes obtained from the qubits from sets $\cZ$ and $\cX$, respectively. Splitting $T$ into two substrings is useful because when Alice has to decide whether to accept or reject the commitment she will only look at one of the substrings (the one measured in the same basis). Clearly, analogous partition applies to $T' = \{T'_{\cZ}, T'_{\cX}\}$.
\subsubsection{No-signalling constraints}
\label{sec:nosign}
Let us think of Alice as talking to Bob and Brian separately and making a separate decision (whether to accept or not) for each of them. We can see that this gives rise to a joint probability distribution with two inputs and two outputs: the inputs are the bits that Bob and Brian were asked by Victor\footnote{We are in the global command model so both Bob and Brian know what they are trying to unveil.} to unveil ($b$ and $b'$, respectively), while the outputs are Alice's binary ($\{$accept, reject$\}$) outcomes (one on each side). We have already defined the maps that Bob and Brian will apply so now we just need to specify what the tests on Alice's side are. As described in the protocol Alice will check whether the relevant substring (determined by the partition into $\cZ$ and $\cX$) is identical to her measurement outcomes and these checks can be expressed as projectors. E.g. if Bob tries to open $b = 0$ ($b = 1$) Alice will apply $\Pi_{B}^{0}$ ($\Pi_{B}^{1})$, where
\begin{align*}
\Pi_{B}^{0} &:= \sum_{s} \ketbra{s}_{S_{\cZ}} \otimes \ketbra{s}_{T_{\cZ}},\\
\Pi_{B}^{1} &:= \sum_{s} \ketbra{s}_{S_{\cX}} \otimes \ketbra{s}_{T_{\cX}}.
\end{align*}
To check Brian's opening she would apply $\Pi_{B'}^{0}$ or $\Pi_{B'}^{1}$, which can be obtained from the projectors above by replacing $T$ with $T'$. Note that the opening maps performed by Bob and Brian and the tests performed by Alice allow us to evaluate the joint probability distribution, which is represented in Table~\ref{tab:no-signalling}.\footnote{The variables that do not appear in our argument have been replaced with placeholders.} As Bob and Brian act on disjoint quantum systems and the tests performed by Alice are classical the probability distribution of outcomes must satisfy no-signalling.
\begin{table}
\centering
\begin{tabular}{ c c c || c | c || c | c ||}
	& & & \multicolumn{4}{|c||}{Alice and Brian}\\
	& & & \multicolumn{2}{|c}{$b' = 0$} & \multicolumn{2}{c||}{$b' = 1$}\\
	& & & accept & reject & reject & accept\\
	\hline \hline \multirow{4}{*}{Alice and Bob} & \multirow{2}{*}{$b = 0$} & accept & $p_{0}$ & $a_{12}$ & $\cdot$ & $ \alpha $\\
  \cline{3-7} & & reject & $a_{21}$ & $a_{22}$ & $a_{23}$ & $a_{24}$\\
  \hhline{~~=====} & \multirow{2}{*}{$b = 1$} & reject & $\cdot$ & $\cdot$ & $\cdot$ & $a_{34}$\\
  \cline{3-7} & & accept & $\cdot$ & $\cdot$ & $\cdot$ & $p_{1}$\\
	\hline \hline
\end{tabular}
\caption{The joint probability distribution describing the two space-like separated openings.}
\label{tab:no-signalling}
\end{table}
%
%
Note that we replaced certain fields ($a_{11}$ and $a_{44}$) by the probability of successfully opening $0$ and $1$ ($p_{0}$ and $p_{1}$), respectively. This follows from the definition of $p_{d}$ in the global command model :
\begin{equation}
\label{eq:pddef}
p_{d} := \Pr[\mbox{accept}, \mbox{accept} | b = d, b' = d].
\end{equation}
Also, we have replaced $a_{14}$ by $\alpha$ because it turns out to be the quantity we will bound in the second part of the proof. The following lemma uses the no-signalling principle to find an upper bound on the sum of $p_{0}$ and $p_{1}$.
\begin{lem}
\label{lem:nosign}
No-signalling between Bob and Brian implies that $p_{0} + p_{1} \leq 1 + \alpha$.
\end{lem}
\begin{proof}
Consider the following no-signalling constraints : $\alpha + a_{24} = a_{34} + p_{1}$ and $a_{21} + a_{22} = a_{23} + a_{24}$. Moreover, we know that each quarter adds up to $1$ so $p_{0} + a_{12} + a_{21} + a_{22} = 1$. Combining the two conditions gives
\begin{equation*}
p_{0} + p_{1} = 1 - a_{12} - a_{21} - a_{22} + \alpha + a_{24} - a_{34} = 1 - a_{12} - a_{23} + \alpha - a_{34} \leq 1 + \alpha. \qedhere
\end{equation*}
\end{proof}
Hence, it is enough to show that as the number of rounds $n$ increases $\alpha$ can be made arbitrarily small, which is the focus of the next section.
\subsubsection{Impossibility of guessing both strings}
\label{sec:impos}
%
%
%
%
%
%
The probability $\alpha$ corresponds to Bob trying to unveil $b = 0$, Brian trying to unveil $b' = 1$ and both openings being accepted. Let $\rho_{S_{\cZ} S_{\cX} T_{\cZ} T_{\cX} T'_{\cZ} T'_{\cX}}$ be the state after all three parties have performed their measurements (note that this state is purely classical)
\begin{equation*}
\rho_{S_{\cZ} S_{\cX} T_{\cZ} T_{\cX} T'_{\cZ} T'_{\cX}} := (\Lambda_{A_{\cZ}}^{0} \otimes \Lambda_{A_{\cX}}^{1} \otimes \Phi_{B}^{0} \otimes \Phi_{B'}^{1}) \rho_{A_{\cZ} A_{\cX} B B'}.
\end{equation*}
As $\alpha$ is the probability that $\rho_{S_{\cZ} S_{\cX} T_{\cZ} T_{\cX} T'_{\cZ} T'_{\cX}}$ passes the relevant tests it can be written as
\begin{equation}
\label{eq:alphaproj}
\alpha = \tr(\Pi_{B}^{0} \Pi_{B'}^{1} \rho_{S_{\cZ} S_{\cX} T_{\cZ} T_{\cX} T'_{\cZ} T'_{\cX}}).
\end{equation}
As operators acting on disjoint subsystems commute we can change the order slightly
\begin{align*}
\alpha &= \tr(\Pi_{B}^{0} \Pi_{B'}^{1} \rho_{S_{\cZ} S_{\cX} T_{\cZ} T_{\cX} T'_{\cZ} T'_{\cX}})\\
&= \tr(\Pi_{B}^{0} \Pi_{B'}^{1} (\Lambda_{A_{\cZ}}^{0} \otimes \Lambda_{A_{\cX}}^{1} \otimes \Phi_{B}^{0} \otimes \Phi_{B'}^{1}) \rho_{A_{\cZ} A_{\cX} B B'})\\
&= \tr(\Pi_{B}^{0} \Pi_{B'}^{1} (\Lambda_{A_{\cX}}^{1} \otimes \Phi_{B'}^{1}) (\Lambda_{A_{\cZ}}^{0} \otimes \Phi_{B}^{0}) \rho_{A_{\cZ} A_{\cX} B B'})\\
&= \tr \left( \Pi_{B'}^{1} (\Lambda_{A_{\cX}}^{1} \otimes \Phi_{B'}^{1}) \left[ \Pi_{B}^{0} (\Lambda_{A_{\cZ}}^{0} \otimes \Phi_{B}^{0}) \rho_{A_{\cZ} A_{\cX} B B'} \right] \right).
\end{align*}
Define
\begin{gather*}
\label{eq:prob}
p := \tr \big(\Pi_{B}^{0} (\Lambda_{A_{\cZ}}^{0} \otimes \Phi_{B}^{0}) \rho_{A_{\cZ} A_{\cX} B B'} \big),\\
\rho_{A_{\cX} T_{\cX} B'}^{\textnormal{pass}} := \frac{1}{p} \tr_{S_{\cZ} T_{\cZ}} \Big[ \Pi_{B}^{0} (\Lambda_{A_{\cZ}}^{0} \otimes \Phi_{B}^{0}) \rho_{A_{\cZ} A_{\cX} B B'} \Big].
\end{gather*}
It is easy to see that $p$ is the probability that Bob passes his test and $\rho_{A_{\cX} T_{\cX} B'}^{\textnormal{pass}}$ is the normalised state conditioned on passing. Hence, $\alpha$ can be written as
\begin{equation}
\label{eq:alphaproj2}
\alpha = \tr \big( \Pi_{B'}^{1} (\Lambda_{A_{\cX}}^{1} \otimes \Phi_{B'}^{1}) \rho_{A_{\cX} T_{\cX} B'}^{\textnormal{pass}} \big) \cdot p.
\end{equation}
This way of writing $\alpha$ allows us to apply Theorem~\ref{thm:uncertainty} to the tri-partite state $\rho_{A_{\cX} T_{\cX} B'}^{\textnormal{pass}}$.

\begin{lem}
\label{lem:alphalemma}
For any strategy adopted by dishonest Bob
\begin{equation}
\label{eq:alphalemma}
\alpha \leq \inf_{\delta \in (0, \frac{1}{2}) } 2^{1 - n(1 - h(\delta))} + 2 \exp \left(- \frac{1}{2} n \delta^{2} \right).
\end{equation}
\end{lem}
\begin{proof}
The trace on the right hand side of~\eqref{eq:alphaproj2} corresponds to the probability that Brian guesses $S_{\cX}$ correctly by applying his opening map on his subsystem \emph{conditioned} on Alice accepting Bob's opening. The guessing probability using a fixed map $\Phi_{B'}^{1}$ is upperbounded by the optimal guessing probability \cite{koenig08} which can be written in terms of the min-entropy. Hence,
\begin{equation}
\label{eq:alpha-minentropy}
\frac{\alpha}{p} = \tr \big( \Pi_{B'}^{1} (\Lambda_{A_{\cX}}^{1} \otimes \Phi_{B'}^{1}) \rho_{A_{\cX} T_{\cX} B'}^{\textnormal{pass}} \big) \leq 2^{- \hmin(S_{\cX} | B')},
\end{equation}
where the min-entropy is evaluated on the state $\rho_{S_{\cX} B'} := \tr_{T_{\cX}} \Lambda_{A_{\cX}}^{1} (\rho_{A_{\cX} T_{\cX} B'}^{\textnormal{pass}})$. To use the uncertainty relation~\eqref{eq:uncertainty} we also need to consider $\rho_{\hat{S}_{\cX} T_{\cX}} := \tr_{B'} \Lambda_{A_{\cX}}^{0} (\rho_{A_{\cX} T_{\cX} B'}^{\textnormal{pass}})$, which would be obtained if Alice decided to make the third (virtual) measurement in a complementary basis. Combining~\eqref{eq:uncertainty} with~\eqref{eq:alpha-minentropy} gives
\begin{equation}
\label{eq:alpha-maxentropy}
\frac{\alpha}{p} \leq 2^{\hmax( \hat{S}_{\cX} | T_{\cX}) - n},
\end{equation}
where $\hmax( \hat{S}_{\cX} | T_{\cX})$ is evaluated on $\rho_{\hat{S}_{\cX} T_{\cX}}$. Note that now we just need to bound the classical conditional max-entropy between two classical random variables (the state $\rho_{\hat{S}_{\cX} T_{\cX}}$ is purely classical). It turns out that it is enough to show that the Hamming distance between $\hat{S}_{\cX}$ and $T_{\cX}$ is small with high probability. To get such a bound we need to examine the (fully classical) state $\rho_{S_{\cZ} \hat{S}_{\cX} T_{\cZ} T_{\cX}} := \tr_{B'} \big[ (\Lambda_{A_{\cZ}}^{0} \otimes \Lambda_{A_{\cX}}^{0} \otimes \Phi_{B}^{0} \otimes \id_{B'}) \rho_{ABB'} \big]$. The fact that $\cZ$ and $\cX$ are random subsets of $[2n]$ allows us to derive the following inequality from the Hoeffding bound~\cite{hoeffding63} (details in Appendix~\ref{app:hoeffder}).
\begin{equation}
\label{eq:hamming}
\textnormal{Pr} \left[ \dham(\hat{S}_{\cX}, T_{\cX}) \geq \delta n \wedge \dham(S_{\cZ}, T_{\cZ}) = 0 \right] \leq \exp \left(- \frac{1}{2} n \delta^{2} \right) =: \varepsilon.
\end{equation}
We can also write it as conditional probability
\begin{equation*}
\textnormal{Pr} \left[ \dham(\hat{S}_{\cX}, T_{\cX}) \geq \delta n | \dham(S_{\cZ}, T_{\cZ}) = 0 \right] \leq \frac{\varepsilon}{p},
\end{equation*}
because $\dham(S_{\cZ}, T_{\cZ}) = 0$ is equivalent to Bob passing the test (and happens with probability $p$ as defined in~\eqref{eq:prob}). Let $0 < \delta < \frac{1}{2}$ and define a binary event, $\Gamma$, such that
\begin{equation*}
\Gamma := 
\begin{cases}
0 \nbox{if}{10} \dham(\hat{S}_{\cX}, T_{\cX}) < \delta n,\\
1 \nbox{if}{10} \dham(\hat{S}_{\cX}, T_{\cX}) \geq \delta n.
\end{cases}
\end{equation*}
If $\Gamma = 0$ then for any particular value of $T_{\cX} = t_{\cX}$ the R{\'e}nyi entropy\footnote{All entropies are evaluated on $\rho_{\hat{S}_{\cX} T_{\cX}}$, except for $\hmin(S_{\cX} | B')$ which is evaluated on $\rho_{S_{\cX} B'}$.} of order $0$ can be bounded by
\begin{equation*}
\textnormal{H}_{0}(\hat{S}_{\cX} | T_{\cX} = t_{\cX}, \Gamma = 0) \leq \log \left( \sum_{i = 0}^{\lfloor n \delta \rfloor} {n \choose i} \right) \leq n h(\delta),
\end{equation*}
where the last inequality comes from a well-known bound (see e.g.\ Lemma 16.19  in~\cite{flum06}). The monotonicity of classical R{\'e}nyi entropies implies that
\begin{equation}
\label{eq:maxentropycond}
\hmax(\hat{S}_{\cX} | T_{\cX} = t_{\cX}, \Gamma = 0) \leq \textnormal{H}_{0}(\hat{S}_{\cX} | T_{\cX} = t_{\cX}, \Gamma = 0).
\end{equation}
If $\Gamma = 1$ then we have no bound better than the maximal value $\hmax(\hat{S}_{\cX} | T_{\cX} = t_{\cX}, \Gamma = 1) \leq n$. It can be shown (see e.g.\ Section 4.3.2 in~\cite{tomamichel12}) that the conditional max-entropy for classical states reduces to
\begin{equation*}
\hmax(Z | Y) := \log \sum_{y \in \cY} \textnormal{Pr}[Y = y] \cdot 2^{\hmax(Z | Y = y)}.
\end{equation*}
As neither of our bounds depends on the particular value of $T_{\cX} = t_{\cX}$, they will not be affected by averaging over all strings $t_{\cX}$. Hence, we only need to average over $\Gamma$
\begin{align}
\label{eq:hmaxfull}
2^{\hmax(\hat{S}_{\cX} | T_{\cX}, \Gamma)} &= \textnormal{Pr}[\Gamma = 0] \cdot 2^{\hmax(\hat{S}_{\cX} | T_{\cX}, \Gamma = 0)} + \textnormal{Pr}[\Gamma = 1] \cdot 2^{\hmax(\hat{S}_{\cX} | T_{\cX}, \Gamma = 1)}\notag \\
&\leq \Big(1 - \frac{\varepsilon}{p} \Big) 2^{n h(\delta)} + \frac{\varepsilon}{p} 2^{n} \leq 2^{n h(\delta)} + \frac{2^{n} \varepsilon}{p}.
\end{align}
One bit of information cannot decrease the entropy by more than 1 bit (see e.g.\ Proposition 5.10 in~\cite{tomamichel12}), hence
\begin{equation}
\label{eq:hmax1bit}
\hmax(\hat{S}_{\cX} | T_{\cX}) \leq \hmax(\hat{S}_{\cX} | T_{\cX}, \Gamma) + 1.
\end{equation}
%
%
%
%
%
%
%
%
%
Hence, from~\eqref{eq:alpha-maxentropy},~\eqref{eq:hmaxfull} and~\eqref{eq:hmax1bit} we get
\begin{equation*}
\alpha \leq 2p \left[ 2^{-n(1 - h(\delta))} + \frac{\varepsilon}{p} \right] \leq 2^{1 - n(1 - h(\delta))} + 2 \exp \left(- \frac{1}{2} n \delta^{2} \right),
\end{equation*}
which directly implies our claim.
\end{proof}
Finally, Theorem \ref{thm:mainresult} follows directly from Lemmas~\ref{lem:nosign} and \ref{lem:alphalemma}.

\section{Conclusions and open questions}
\label{sec:conclusions}
Our interest in bit commitment protocols based on the relativistic constraint was sparked by recent papers by Kent~\cite{kent11, kent12a}. While the author gave an intuition for the security of the protocol based on BB84 states, no explicit security bounds were given. Once we had proven the security of the protocol and calculated such bounds, we became interested in other split models: which of them can give us security and in which of them are quantum protocols more powerful than classical ones? We have investigated the minimal split assumptions that might allow for secure bit commitment and we have shown that they are indeed sufficient. We have found that in the $\beta$-split under the global command quantum protocols are more powerful than classical ones.

We have proven security of bit commitment with respect to the weakly binding definition, which is non-composable. We also know that the usual stronger definition (which would imply composability) is not achievable. We cannot hope for universal composability but maybe it is possible to prove some weaker form of composability. For example, is it possible to combine $n$ bit commitment protocols~\cite{kent12a} to obtain a secure string commitment scheme? If it is not secure one might investigate if there are some extra constraints (e.g.~that the commit phases are executed sequentially or that the unveilings happen simultaneously at space-like separated points) that would guarantee composability.

One might also wonder whether these models allow us to construct other cryptographic primitives. Probably the most natural one to look at would be oblivious transfer \cite{rabin81, wullschleger06}. Unfortunately, the primitive of oblivious transfer requires the security to last forever. This would only be possible if certain parties remained split forever, which cannot be motivated by relativistic assumptions. Moreover, if certain parties were to remain split forever then oblivious transfer can be implemented even classically~\cite{naor00}. It is possible, however, that some weaker form of oblivious transfer (in which the security does not last forever) can be proven secure in relativistic models.
\appendices
\section{Hoeffding bound}
\label{app:hoeffder}
In Lemma~\ref{eq:alphalemma} we need to bound the probability that sampling a small, random substring gives rise to the statistics which is very different from the true statistics of the entire string. The Hoeffding bound is exactly the tool we need. Suppose that we have a string of length $2n$ which contains $n_{\textnormal{err}}$ errors and let $\overline{\Lambda} = \frac{n_{\textnormal{err}}}{2n}$ denote the error fraction in the whole string. Let us take a random sample of the string of size $k$ and denote the error fraction in the sample by $\lambda$. Then, the Hoeffding bound \cite{hoeffding63} states that
\begin{equation*}
{\rm Pr} \left[ \overline{\Lambda} \geq \lambda + \frac{\delta}{2} \right] \leq \exp \left(- \frac{1}{2} k \delta^{2} \right).
\end{equation*}
Adding an extra event cannot increase the probability
\begin{equation*}
{\rm Pr} \left[ \overline{\Lambda} \geq \lambda + \frac{\delta}{2} \wedge \lambda = 0 \right] \leq \exp \left(- \frac{1}{2} k \delta^{2} \right).
\end{equation*}
The expression inside the square bracket can be rewritten, giving us
\begin{equation*}
{\rm Pr} \left[ n_{\textnormal{err}} \geq \delta n \wedge \lambda = 0 \right] \leq \exp \left(- \frac{1}{2} k \delta^{2} \right).
\end{equation*}
This is exactly the bound we use in~\eqref{eq:hamming}.
\section{Composability issues}
For the sake of completeness we state some observations concerning composability. On one hand we show that the weak bindingness definition is not composable (by giving an explicit counter-example). On the other hand we argue that the usual stronger definition~\cite{damgaard07} cannot be satisfied in the split setting.
\subsection{Counter-example to the composability of the weakly binding definition}
\label{app:counterexample}
In Section \ref{sec:bitcom} we explained what it means that a bit commitment protocol is weakly binding and we also said that the definition does not guarantee composability, e.g.\ executing the protocol $n$ times does not necessarily give a secure string commitment (string commitment is an extension of bit commitment in which we are allowed to commit to a bitstring of length $n$ rather than just a single bit). Let us explain what the source of the problem is. Consider a bit commitment protocol which is binding in the sense that with probability $\frac{1}{2}$ Bob can unveil either bit successfully and with probability $\frac{1}{2}$ he will fail regardless of his intentions. Clearly, we would \emph{not} call this protocol secure. However, as $p_{0} = p_{1} = \frac{1}{2}$ it satisfies the $\varepsilon$-weakly binding definition for $\varepsilon = 0$. To expose the problem even further consider the task of string commitment. Analogous to the bit commitment case suppose that at the end of the commit phase Alice and Bob share a state $\rho_{AB}$. Let $q_{s}(\rho_{AB})$ be the probability that Bob successfully unveils string $s$. Then it is natural to say that a string commitment protocol is $\delta$-weakly binding if for all states $\rho_{AB}$ it satisfies
\begin{equation*}
\sum_{s} q_{s}(\rho_{AB}) \leq 1 + \delta.
\end{equation*}
Now consider a string commitment protocol such that Alice with probability $\frac{1}{2}$ accepts anything while with probability $\frac{1}{2}$ rejects everything. It is clear that this is not a secure string commitment box as $\sum_{s} q_{s}(\rho_{AB}) = \frac{1}{2} 2^{n} = 2^{n - 1}$. However, if we look at each bit separately we will find that $p_{0} = p_{1} = \frac{1}{2}$ and so each bit commitment is weakly binding. This shows that combining $n$ weakly binding bit commitments does not imply that the resulting string commitment is secure.
\subsection{Impossibility of satisfying the stronger definition}
\label{sec:strongerdefimpossible}
\begin{df}
\label{df:epsbnd}
\cite{damgaard07} A bit commitment protocol is \emph{$\epsilon$-binding} if the fact that Alice is honest ensures that for any state at the beginning of the open phase, $\rho_{AB}$, there exists an extension of the form
\begin{equation*}
\rho_{ABD} = P_{D}(0) \ketbra{0}_{D} \otimes \rho_{AB}^{0} + P_{D}(1) \ketbra{1}_{D} \otimes \rho_{AB}^{1},
\end{equation*}
where $D$ is a classical register and $P_{D}$ is a probability distribution, for which the conditioned states satisfy $p_{1 - b}(\rho_{AB}^{b}) \leq \epsilon$ for $b \in \{0, 1\}$.
\end{df}
While this definition has proven useful in the bounded and noisy storage models~\cite{damgaard08, wehner08} we argue that it is generally inapplicable outside of these scenarios. The security in these models results from the fact that Alice and Bob cannot purify the protocol, as there is a subsystem, referred to as the environment, $E$, which they do not have access to. In other words $\rho_{AB}$ is not pure because we trace out the environment $E$, e.g. a pure state $\ket{\phi}_{ABE}$ leads to $\rho_{AB} = \tr_{E} \ketbra{\phi}_{ABE}$. The following argument shows that if the model does not prevent the parties from purifying the protocol then Definition~\ref{df:epsbnd} can only be satisfied for $\epsilon \geq \frac{1}{2}$. Suppose that Bob commits to an equal superposition of $0$ and $1$ (as explained above). If Alice and Bob start in a pure state and execute a purified version of the protocol (i.e.\ implement all operations as unitaries, generate coherent randomness and keep all the measurements quantum) then the state at the beginning of the open phase is pure. One possible opening strategy for Bob is to measure the control qubit, which collapses the state. The collapsed state is exactly \emph{as if} Bob had generated a random bit $b$ at the very beginning of the protocol and honestly committed to it. Such a strategy gives us a lower bound on how well Bob can open each bit, namely $p_{b}(\rho_{AB}) \geq \frac{1}{2}$ for $b \in \{0, 1\}$. As the overall state is pure at the beginning of the open phase, any classical register $D$ must necessarily be independent, which means that $\rho_{AB}^{0} = \rho_{AB}^{1} = \rho_{AB}$. Then $p_{1}(\rho_{AB}^{0}) = p_{1}(\rho_{AB}) \geq \frac{1}{2}$ so Definition~\ref{df:epsbnd} can only hold for $\epsilon \geq \frac{1}{2}$. This argument shows that Definition~\ref{df:epsbnd} cannot be satisfied by protocols that do not assume the presence of some external system inaccessible to either party.

\section{Guaranteed message delivery time models}
\label{app:gmdt}
Suppose that Bob, based on Earth, exchanges messages with Alice, who is on the Moon. Special relativity states that no message can travel faster than the speed of light, hence the minimum delivery time equals about $1.26 s$. This scenario motivates the study of models in which there are two separated sites and while intra-site communication can be instantaneous, any inter-site message takes at least $\Delta t$ to be delivered. We also assume that the inter-site (classical or quantum) channels are perfectly secure (neither party can read or alter anything that is on the wire).
\subsection{One agent per site}
\begin{figure}[h]
\centering
\begin{tikzpicture}[scale=1]
	\draw[fill=gray!30] (0, 0) rectangle (2, 1);
	\draw[fill=gray!30] (6, 0) rectangle (8, 1);
	\node (Bob) at (1, 0.5) {Bob};
	\node (Alice) at (7, 0.5) {Alice};
	\draw [thick, ->] (2.2, 0.65) -- (5.8, 0.65);
	\draw [thick, ->] (5.8, 0.35) -- (2.2, 0.35);
	\node (deltat) at (4, 1) {$\Delta t$};
\end{tikzpicture}
\caption{The simplest guaranteed message delivery time model: one agent per site.}
\label{fig:one-agent}
\end{figure}
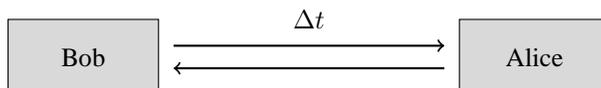
The simplest model (illustrated in Figure~\ref{fig:one-agent}) assumes that each party controls one site. Clearly, if Bob sends a bit $b$ to Alice he is committed to it. The commitment is perfect because at time $t = 0$ Bob is fully committed (he cannot alter his commitment at any later time), while at the same time until time $t = \Delta t$ Alice is fully ignorant about the commitment. The drawback of such a scheme is the fact that the commitment only lasts for $\Delta t$ and then automatically opens. Such schemes have been studied before~\cite{boneh00} but in a slightly different context. The conclusion is that for certain applications (e.g.\ constructing a strong coin flip, signing contracts) such \emph{timed commitments} are good enough, while for others (e.g.~Yao's construction of OT using quantum communication~\cite{yao95, wullschleger06}) they are not. To illustrate the limitations of this model let us consider if it is possible to construct a commitment that lasts for longer that $\Delta t$. Classically, this is not possible and the intuitive argument is simple. In the absence of noise classical protocols are fully deterministic and no probabilities can arise. For each of the bits Bob either \emph{can} ($p_{b} = 1$) or \emph{cannot} ($p_{b} = 0$) unveil it. Hence, the distinction between being and not being committed is sharp (either $p_{0} + p_{1} = 2$ or $p_{0} + p_{1} = 1$). Bob being committed implies that the information beyond his control determines the bit. As Alice will have received all the messages in transit after time at most $\Delta t$ she will be able to learn the committed bit. Therefore, no commitment can be made longer than $\Delta t$. In the quantum world the situation is more complicated due to two things. First of all, quantum mechanics is a probabilistic theory so there is no sharp distinction between being and not being committed\,---\,Bob can be partially committed. The second complication is the no-cloning theorem. Suppose that at some point Bob becomes, to some extent, committed, which means that the information on Alice's side combined with the messages on the wire give away some information about his commitment. Now, assume that Alice waits until the messages arrive (at most $\Delta t$) and does some measurements to learn something about Bob's commitment. Clearly, the standard hiding-binding trade-off applies. However, the honest protocol might require Alice to return some states to Bob before the messages arrive and so by keeping them she takes a risk of being caught cheating. It is an open question if this time-constrained scenario gives us some advantage over the standard scenario for constructing cheat-sensitive bit commitments. It is clear, however, that no secure (hiding) bit commitment can last longer than $\Delta t$. Hence, for this specific purpose quantum and classical protocols are equally powerful.
\subsection{Two agents per site}
\begin{figure}[h]
\centering
\begin{tikzpicture}[scale=1]
	\draw[fill=gray!30] (0, 0) rectangle (2, 1);
	\draw[fill=gray!30] (6, 0) rectangle (8, 1);
	\draw[fill=gray!30] (0, 1.5) rectangle (2, 2.5);
	\draw[fill=gray!30] (6, 1.5) rectangle (8, 2.5);
	\node (alice) at (1, 2) {Alice};
	\node (Amy) at (7, 2) {Amy};
	\node (bob) at (1, 0.5) {Bob};
	\node (Brian) at (7, 0.5) {Brian};
	\draw [thick, ->] (2.2, 2.2) -- (5.8, 2.2);
	\draw [thick, ->] (2.2, 2) -- (5.8, 0.8);
	\draw [thick, ->] (5.8, 2) -- (2.2, 0.8);
	\draw [thick, ->] (5.8, 0.5) -- (2.2, 0.5);
	\node (deltat) at (4, 2.5) {$\Delta t$};
\end{tikzpicture}
\caption{A more complicated guaranteed message delivery time model: two agents per site.}
\label{fig:two-agents}
\end{figure}
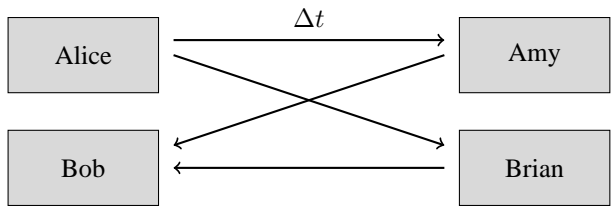
A slighty more complicated model (illustrated in Figure~\ref{fig:two-agents}) assumes that each party has a trusted agent at each site (Bob trusts his agent Brian and Alice trusts her agent Amy). Protocols implementing bit commitment in such a scenario, in which the commitment can be sustained indefinitely as long as messages are exchanged at each site have been presented in~\cite{kent99, kent05}. After the exchange stops the commitment remains valid for $\Delta t$ and then expires. These protocols have been shown to be secure against classical attacks and are conjectured to be secure against any quantum attack.
\section{Classical protocols against quantum adversaries}
Some of the protocols we present are purely classical but in order to determine whether they are secure against quantum adversaries we need to translate them into the quantum formalism. This section describes briefly how this can be achieved and analyses the security of these protocols in the quantum setting. While the actual security proofs may appear trivial, we have decided to include them for completeness.
\subsection{Classical protocol in the quantum formalism}
Sending a classical bit $b \in \{0, 1\}$ is equivalent to encoding it in the computational basis and sending the resulting state $\ket{b}$ to the other party. Receiving a classical bit corresponds to receiving a qubit and immediately measuring it in the computational basis.
\subsection{Bit commitment from secret sharing}
\label{app:secshasec}
Here we analyse Protocol \ref{prot:secshar} from Section \ref{sec:alphasep}. If Alice and Amy are honest they will measure the qubits they receive immediately in the computational basis. Once the measurement outcomes are known Bob's commitment is well-defined and he will not be able to cheat. If Bob is honest $r$ will be a truly random bit. Then what Alice and Amy receive can be described by the following density matrix
\begin{equation*}
\rho_{AA'}^{d} = \frac{1}{2} [ \ketbra{0}_{A} \otimes \ketbra{d}_{A'} + \ketbra{1}_{A} \otimes \ketbra{1 - d}_{A'}].
\end{equation*}
It is easy to convince ourselves that while $\rho_{AA'}^{0}$ and $\rho_{AA'}^{1}$ are perfectly distinguishable the reduced states are fully mixed, $\rho_{A}^{0} = \rho_{A}^{1} = \rho_{A'}^{0} = \rho_{A'}^{1} = \frac{\mathbb{I}}{2}$. Hence, Alice and Amy remain perfectly ignorant about Bob's commitment as long as they are separated.
\subsection{Bit commitment in the local command}
\label{app:loccomsec}
Here we analyse Protocol \ref{prot:trivlocal} from Section \ref{sec:bobsplitopen}. Clearly, the protocol is perfectly hiding because Alice does not receive any messages until the beginning of the open phase. To show that it is also weakly binding we need to employ no-signalling between Bob and Brian.
\begin{lem}
Protocol~\ref{prot:trivlocal} is weakly binding with $\varepsilon = 0$.
\end{lem}
\begin{proof}
Suppose that Bob and Brian want to cheat. At the beginning of the open phase each of them picks an opening strategy from sets $R$ and $S$, respectively. Note that this has to be done independently because they are not allowed to communicate. Bob receives the command so his distribution will in general depend on the command and if the command is $b$ denote the probability of picking $r \in R$ by $p_{R}^{b}(r)$. For the second player the distribution has to be fixed and the probability of picking $s \in \cS$ equals $p_{S}(s)$, regardless of what the value of $b$ is. Recall from Section~\ref{sec:bitcom} that $p_{b}$ is the probability that Alice accepts the commitment if the command is $b$. Hence, we can write
\begin{equation*}
p_{b} = \sum_{r \in R} \sum_{s \in \cS} p_{R}^{b}(r) p_{S}(s) p(x = b, y = b | r, s) \leq \sum_{r \in R} \sum_{s \in \cS} p_{R}^{b}(r) p_{S}(s) p(y = b | r, s).
\end{equation*}
By no-signalling we know that $p(y = b | r, s)$ does not depend on $r$ so we can write $p(y = b | s)$ instead. Then we get
\begin{gather*}
p_{0} + p_{1} \leq \sum_{r \in R} \sum_{s \in \cS} \Big[ p_{R}^{0}(r) p_{S}(s) p(y = 0 | s) + p_{R}^{1}(r) p_{S}(s) p(y = 1 | s) \Big] =\\
\sum_{s \in \cS} p_{S}(s) \Big[ p(y = 0 | s) + p(y = 1 | s) \Big] = 1. \qedhere
\end{gather*}
\end{proof}
One might also wonder whether the protocol satisfies the stronger binding requirement (Definition~\ref{df:epsbnd}). However, a similar argument to the one sketched out in Section~\ref{sec:secalice} shows that the stronger definition cannot hold.
\section*{Acknowledgements:}
We thank Roger Colbeck, Fabian Furrer, Tomasz Paterek and Severin Winkler for helpful discussions 
and Adrian Kent for valuable comments on an earlier version of this manuscript.
\bibliographystyle{IEEEtran}
\bibliography{/home/jedrek/Dropbox/Latex/librarysan}
\begin{IEEEbiographynophoto}{J\k{e}drzej Kaniewski}
was born on November 11, 1987 in Warsaw, Poland. He studied at Cambridge University and graduated in 2011 with a B.A.~(Natural Sciences) and MMath degrees. He is currently a Ph.D.~student at the Centre for Quantum Technologies, Singapore.
\end{IEEEbiographynophoto}
\begin{IEEEbiographynophoto}{Marco Tomamichel}
was born on March 13, 1981, in St. Gallen (Switzerland). He studied Electrical Engineering at ETH Zurich (Switzerland), where he graduated in 2007 with a M.Sc.~in Electrical Engineering and Information Technology degree. He then graduated with a Ph.D.~at the Institute of Theoretical Physics at ETH Zurich (Switzerland). He is currently working as a post-doctoral research fellow at the Centre for Quantum Technologies at the National University of Singapore. His research interests include classical and quantum information theory in the non-asymptotic regime, and its applications to cryptography.
\end{IEEEbiographynophoto}
\begin{IEEEbiographynophoto}{Esther H{\"a}nggi}
studied at EPF Lausanne (Switzerland) and graduated in 2005 with a Master's degree in Physics. In 2010
she obtained a doctoral degree in Computer Science from ETH Zurich (Switzerland), before moving to the Centre for Quantum Technologies at the National University of Singapore to work as a post-doctoral research fellow. Her research interests are in the area of quantum information and (quantum) cryptography.

\end{IEEEbiographynophoto}
\begin{IEEEbiographynophoto}{Stephanie Wehner}
is a physicist and computer scientist at the Centre for Quantum Technologies, National University of Singapore, born in Wuerzburg, Germany. She studied at the University of Amsterdam and obtained her Ph.D.~at CWI, before moving to Caltech as a post-doctoral researcher under John Preskill. Since 2010 she is an assistant professor in the Department of Computer Science at the National University of Singapore and a Principal Investigator at the Centre for Quantum Technologies.
\end{IEEEbiographynophoto}
\end{document}